\documentclass[12pt, a4paper ]{amsart}
\usepackage[all]{xy} 
\usepackage{amsfonts}
\usepackage{amssymb}
\usepackage{mathrsfs}
 
\theoremstyle{plain}                    
\newtheorem{theorem}{Theorem}[section]      
\newtheorem{proposition}[theorem]{Proposition}    
\newtheorem{lemma}[theorem]{Lemma}            
\theoremstyle{definition}               
\newtheorem{definition}[theorem]{Definition}        
\newtheorem{example}[theorem]{Example}      
\theoremstyle{remark}                   
\newtheorem{remark}[theorem]{Remark}           

\def\cE{\mathcal{E}}
\def\cF{\mathcal{F}}
\def\cD{\mathcal{D}}

\def\Res{\mathop{Res}}

\def\d{\partial}
\def\CP{\mathbb{C}\mathrm{P}}
\def\p{\partial}
\def\normordboson{ \mathop{\genfrac{}{}{0pt}3{*}{*} }}
\def\normordfermion{ \mathop{\genfrac{}{}{0pt}3{\bullet}{\bullet} }}
\def\pv{\mathop{\mathsf{egv}}\nolimits}
\def\Tr{\mathop{\mathrm{Tr}}}

\def\Res{\mathop{\mathrm{Res}}}
\def\N{\mathbb{N}} 
\def\Q{\mathbb{Q}}
\def\Z{\mathbb{Z}}
\def\J{\mathcal{J}}

\newcommand{\cor}[1]{\langle 0 |\, #1 \, | 0 \rangle}

\renewcommand{\beta}{\hbar}

\begin{document}
\title[Hurwitz theory, KP integrability and quantum curves]{Ramifications of Hurwitz theory, KP integrability and quantum curves}
\author{A.~Alexandrov}

\address{A.~A.: Centre de recherches math\'ematiques, Universit\'e de Montr\'eal, P.~O.~Box 6128, Centre-ville Station, Montr\'eal, Qu\'ebec H3C 3J7,
	Canada; Department of Mathematics and Statistics,
Concordia University,
1455 de Maisonneuve Blvd. W.
Montreal, Quebec H3G 1M8, Canada; and Institute for Theoretical and Experimental Physics, 25 Bolshaya Cheryomushkinskaya
Ulitsa, Moscow 117218, Russia}
\email{alexandrovsash@gmail.com}

\author{D.~Lewanski}

\address{D.~L.: Korteweg-de Vries Institute for Mathematics, University of Amsterdam, Postbus 94248, 1090 GE Amsterdam, The Netherlands}
\email{D.Lewanski@uva.nl}

\author{S.~Shadrin}

\address{S.~S.: Korteweg-de Vries Institute for Mathematics, University of Amsterdam, Postbus 94248, 1090 GE Amsterdam, The Netherlands}
\email{S.Shadrin@uva.nl}

\vspace{-10.0cm}

\begin{center}
\hfill ITEP/TH-19/15
\end{center}

\vspace{1.0cm}

\begin{abstract}
	In this paper we revisit several recent results on monotone and strictly monotone Hurwitz numbers, providing new proofs. In particular, we use various versions of these numbers to discuss methods of derivation of quantum spectral curves from the point of view of KP integrability and derive new examples of quantum curves for the families of double Hurwitz numbers. 
\end{abstract}

\maketitle

\tableofcontents

\section{Introduction} 

\subsection{Hurwitz numbers}
The purpose of this paper is to survey some variations of the concept of Hurwitz numbers and their generating functions. Recall that a simple Hurwitz number $h_{g,\mu}$ depends on a genus $g\geq 0$ and a partition $\mu\vdash d$ of length $\ell=\ell(\mu)$, $\mu=(\mu_1\geq \cdots \geq \mu_\ell)$, $\sum_{i=1}^\ell \mu_i = d$. By definition, $h_{g,\mu}$ is the weighted number of ramified coverings of a sphere $\CP^1$ by a genus $g$ surface, whose degree is $d$, whose monodromy near $\infty\in\CP^1$ is a permutation of a cyclic type $\mu$, and these coverings must have simple ramification points over fixed $2g-2+n+d$ points in $\CP^1\setminus \{\infty\}$.   

These numbers satisfy plenty of interesting properties, and for our paper the most important ones are
\begin{itemize}
\item[$\bullet$] The generating function of Hurwitz numbers 
\begin{equation}
Z(\mathbf{p},\beta):=\exp\left( \sum_{g,\mu} h_{g,\mu} p_\mu \frac{\beta^{2g-2+\ell(\mu)+|\mu|}}{(2g-2+\ell(\mu)+|\mu|)!}\right)
\end{equation}
is a tau-function of the KP hierarchy~\cite{Okounkov2000,Kazarian2009}.
\item[$\bullet$] The principal specialization $\Psi(x,\hbar)$ of the generating function satisfies a differential equation
\begin{equation}
\left(\hat y - \hat x e^{\hat y}\right) \Psi(x,\hbar)=0, \qquad \hat x = x\cdot, \hat y = \hbar x\frac{\p}{\p x},
\end{equation}
 called quantum curve~\cite{Zhou2012}. Here by principle specialization we call the substitution $p_\mu=(x/\hbar)^{|\mu|}$ in the formula for $Z$ above.
\item[$\bullet$] Hurwitz numbers for fixed  $g\geq 0$ and $\ell\geq 1$ can be arranged into the so-called $\ell$-point functions, whose differentials satisfy the topological recursion in the sense of Chekhov--Eynard--Orantin for the spectral curve
\begin{equation}
\log x=\log y-y,
\end{equation}
see~\cite{EynardMulaseSafnuk2011}.
\item[$\bullet$] There is a formula for $h_{g,\mu}$ in terms of the intersection numbers on the moduli space of curves $\overline{\mathcal{M}}_{g,\ell}$:
\begin{equation}
h_{g.\mu}=
 \frac{(2g-2+\ell+d)!}{|\mathrm{Aut}(\mu)|} 
 \prod_{i=1}^{\ell(\mu)} \frac{\mu_i^{\mu_i}}{\mu_i !} 
 \int_{\overline{\mathcal{M}}_{g,\ell}}
 \frac{1-\lambda_1+\cdots\pm\lambda_g}{\prod_{i=1}^{\ell(\mu)} (1-\mu_i\psi_i)}
\end{equation}
(the ELSV formula)~\cite{ELSV,DuninKazarian2015}.
\end{itemize}

These results are related to each other, and it is interesting to specify a class of combinatorial problems depending in a natural way on a genus parameter $g\geq 0$ and a partition $\mu$, where the same sequence of results can be derived. Let us explain why we find this sequence of results important.

\subsection{Outline of the logic} \label{sec:outlineoflogic}
Let us assume that we start with a combinatorial problem depending on a parameter $g\geq 0$ and a partition $\mu$, and its generating function appears to be a KP tau-function. Then we have the following:

\subsubsection*{Step 1: From KP to quantum curve}
In the case when the generating function of some problem of enumerative geometry can be identified with a KP tau-function, the integrable hierarchy often allows us to find a quantum spectral curve. 
Indeed, the principal specialization of the generating function coincides with the so-called first basis vector of the corresponding point of the Sato Grassmannian, and, as it was observed in \cite{AlexandrovTauFunctions2015} (see also~\cite{Schwarz2015,Zhou2015,JianZhou2015-2}), this reduces the problem to a specialization of a suitable Kac--Schwarz operator that would annihilate it.

\subsubsection*{Step 2: From quantum curve to topological recursion} Once we have a quantum curve, we can formulate a precise conjecture that the differentials of the $\ell$-point functions satisfy the topological recursion~\cite{EynardOrantin2007} for the spectral curve obtained by the dequantization of the quantum curve. This relation was made explicit in~\cite{GukovSulkowski2014}. Note that the spectral curve should also correspond to the $(g,\ell)=(0,1)$ part of  the problem~\cite{DumitrescuMulaseSafnukSorkin2013}, and this property is automatically implied by the quantum curve, see e.g.~\cite{MulaseShadrinSpitz2013,DuninBarkowskiMulaseNoburyPopolitovShadrin2013}.  

\subsubsection*{Step 3: From topological recursion to intersection numbers} Once we have spectral curve topological recursion, we can immediately conclude that the corresponding combinatorial problem can be solved in terms of some intersection numbers on the moduli space of curves that represent the correlators of a semi-simple cohomological field theory with a possibly non-flat unit~\cite{EynardIntersections1,EynardInvariantsModuli2,DOSS}, and therefore, have expressions in terms of the Givental graphs~\cite{DSS,DOSS}.

\subsubsection*{Discussion of Steps 1-3} The most important point of this sequence of steps is that Step 1 provides us with a conjectural spectral curve for Step 2 and, therefore, with a conjectural intersection number formula in Step 3. Thus, analysis of the principal specialization in the framework of the KP integrability appears to be a powerful tool that provides very precise conjectural links between combinatorial problems and the intersection theory of the moduli space of curves. 

This logic allows one to prove the ELSV-type formulas in some cases, for instance, this way the original ELSV is proved in~\cite{DuninKazarian2015}, the Johnson--Pandharipande--Tseng formula for the orbifold Hurwitz numbers is proved in~\cite{DuninLewanski2015}, and the conjectural ELSV-type formula for the $r$-spin Hurwitz numbers is derived, in a new way, in~\cite{ShadrinSpitzZvonkine2015}. The corresponding quantum curves (that might be considered as the sources of all these formulas) are derived in~\cite{Zhou2012,MulaseShadrinSpitz2013}. In all these examples, however, the ELSV-type formulas were known before, without any relation to spectral/quantum curves. 

\subsection{Results in this paper}
Rather general models of Hurwitz type are known to be described by the KP/Toda tau-functions \cite{AMMN2012,OP02,HarnadOrlov2015}, thus, the logic that we outline above can be applied to them.
In this paper we focus on the first step for a number of Hurwitz-type theories based on the symmetric functions of the Jucys-Murphy elements in the group algebra of the symmetric group. These theories were considered recently in connection to enumeration of dessins d'enfants~\cite{AlexandrovMironovMorozovNatanzon2014}, expansion of hypegeometric tau-functions~\cite{HarnadOrlov2015}, study of the HCIZ matrix model~\cite{GouldenGuayNovak2014}, and topological recursion~\cite{DoDyerMathews2014,DoKarev2015}.

We revisit with new proofs a number of results in~\cite{AlexandrovMironovMorozovNatanzon2014,HarnadOrlov2015,DoKarev2015}, namely, 
\begin{itemize}
\item[--] we establish relations between various geometric interpretations for these Hurwitz-type theories;
\item[--] we provide the group operators that generate the corresponding tau-functions;
\item[--] we derive the quantum curves from the Kac--Schwarz operators.
\end{itemize}

Once we have a quantum curve, we can immediately produce an ELSV-type formula. We give a detailed computation for the monotone Hurwitz numbers --- this answers a question posed in~\cite{GouldenGuayNovak2014}, and, in fact, it is not a conjecture but a theorem since the corresponding Step 2 (a proof of the topological recursion) was derived in~\cite{DoDyerMathews2014}.

The description of the 2D Toda tau-function for the double monotone Hurwitz numbers in terms of the KP Kac--Schwarz operators allows us to construct the quantum spectral curve for this case. The second set of the Toda times plays the role of linear parameters of the corresponding operator. We use this Kac--Schwarz description in order to derive a system of linear differential operators that annihilate the tau-function for the double monotone Hurwitz numbers and uniquely characterize it. 

In addition, we derive a number of new quantum curves for similar Hurwitz theories. In particular, this yields an interesting example for which we can say in advance that the logic outlined above does not apply. Namely, we have an example where the dequantization of the quantum curve doesn't give a spectral curve suitable for the corresponding topological recursion. 

\subsection{Organization of the paper} In Section~\ref{sec:KPhierarchy} we briefly recall the necessary facts from the theory of the KP hierarchy. In Section~\ref{sec:Jucys} we recall the necessary facts from the Jucys theory. In Section~\ref{sec:HurwitzTheory} we define a variety of Hurwitz-type problems that we study in this paper, and explain the correspondences between them that follow from the Jucys correspondence.  In Section~\ref{sec:OperatorsForBlocks} we embed these Hurwitz-type problems in the framework of the KP formalism. Section~\ref{sec:MonotoneHurwitz} is devoted to the study of the monotone Hurwitz numbers. We derive in a new way a quantum curve for them, compute the associated ELSV-type formula, and provide the linear constrains for the tau-function of the double monotone Hurwitz numbers. Finally, in Section~\ref{sec:FurtherExamples} we derive quantum curves for some further examples that are interesting from various points of view (in particular, the one whose classical limit does not give a proper spectral curve). 

\subsection{Acknowledgments} We thank A.~Popolitov for fruitful discussions and careful reading of the first draft of this paper. We also thank A. Orlov for an interesting discussion and the referee for helpful suggestions.  

A.~A. supported by
the Natural Sciences and Engineering Research Council of Canada (NSERC), by the Fonds derecherche du Qu\'ebec Ð Nature et technologies (FRQNT), by RFBR grant 15-01-04217 and joint RFBR grant 15-52-50041-YaF. A.~A. acknowledges
hospitality of KdV Institute for Mathematics, IBS Center for Geometry and Physics, and IHES.

D.~L. and S.~S. we supported by the Netherlands Organization for Scientific Research.


\section{KP hierarchy and Kac-Schwarz operators}\label{S1} \label{sec:KPhierarchy}
In this section we give a brief recollection of some of the basic concepts of KP integrability used in this paper. For more details see, e.g., \cite{SatoSato1983,SegalWilson1985,FukumaKawaiNakayama1992,MiwaJimboDateBook2000,AlexandrovTauFunctions2015} and references therein. 

The KP hierarchy can be described by the bilinear identity satisfied by the tau-function $\tau({\bf t})$, namely  
\begin{equation}\label{bi1}
\oint_{{\infty}} e^{\xi ({\bf t}-{\bf t'},z)}
\,\tau ({\bf t}-[z^{-1}])\,\tau ({\bf t'}+[z^{-1}])dz =0,
\end{equation}
where $\xi({\bf t},z)=\sum_{k=1}^\infty t_k z^{k}$ and we use the standard notation
\begin{equation}\label{shiftedt}
{\bf t} \pm \left[z^{-1}\right]=\left\{t_1\pm\frac{1}{z},t_2\pm\frac{1}{2z^2},t_3\pm\frac{1}{3z^3},\dots\right\}.
\end{equation}

In Hurwitz-type problems it is often convenient to work in the coordinates $\mathbf{p}$ instead of $\mathbf{t}$, where $p_k=kt_k$, $k=1,2,\dots$.

\subsection{Semi-infinite wedge space}
We consider the vector space $V:=\bigoplus_{c\in\Z} V_c$ spanned by the vectors that are obtained from
$$
|0\rangle := z^0 \wedge z^{-1} \wedge z^{-2} \wedge \cdots 
$$
by applying a finite number of the operators $\psi_i:=z^i \wedge$ and $\psi_i^*:= \frac{\d}{\d (z^i)}$, $i\in\mathbb{Z}$. The gradation $c$ is introduced as follows: 
$$
|0\rangle\in V_0,\ \deg\psi_i = 1,\ \deg \psi_i^*=-1,\ i\in\Z. 
$$
In particular, the vector space $V_0$ has a basis that consists of the vectors 
$$
v_\lambda := z^{\lambda_1-0}\wedge z^{\lambda_2-1} \wedge z^{\lambda_3-2} \wedge \cdots,
$$
where $\lambda=(\lambda_1\geq \lambda_2\geq\cdots\geq \lambda_{\ell(\lambda)}\geq 0\geq 0\geq \cdots)$ is a Young diagram. Note that $|0\rangle = v_\emptyset$. 

We define the operator $\normordfermion\psi_i\psi_j^*\normordfermion$ on $V_0$ to be $\psi_i\psi_j^*$ if $j>0$ and $-\psi_j^*\psi_i$ if $j\leq 0$. The map $E_{ij} \to \normordfermion\psi_i\psi_j^*\normordfermion$ gives a projective representation of $\mathfrak{gl}_\infty$ in $V_0$. 

Consider the operators $\alpha_n:=\sum_{i\in \Z} \normordfermion\psi_{i-n}\psi_i^*\normordfermion$ defined on $V_0$. Note that $\alpha_0 V_0 = 0$.

There is a map from $V_0$ to $\mathbb{C}[[\mathbf{t}]]$ given by 
\begin{equation}\label{eq:correspondenceKP}
V_0\ni v \mapsto \langle 0 | \exp(\sum_{i=1}^\infty t_i \alpha_i) v,
\end{equation}
where $ \langle 0 |$ is the covacuum, that is, the covector that returns the coefficient of  $|0\rangle$. For instance, the function that corresponds to $v_\lambda$ is the Schur function $s_\lambda(\mathbf{t})$. 

The description of the tau-functions of the KP hierarchy in this language is the following: the tau-functions correspond to the vectors that belong to the image of the Pl\"ucker embedding of the semi-infinite Grassmannian, also called Sato Grassmannian. On the open cell this means that we are looking for the vectors representable as 
$$\Phi_1\wedge \Phi_2\wedge \Phi_3 \wedge \cdots,$$
where $\Phi_k (z)=z^{1-k}+\sum_{m=2-k}^{\infty}\Phi_{km} z^{m}$, $\Phi_{km}\in\mathbb{C}$,  are known as basis vectors. This description immediately implies that the group $GL(V_0)$ is the group of symmetries of the KP hierarchy. 

The map~\eqref{eq:correspondenceKP} allows to translate the infinitesimal symmetries of the semi-infinite Grassmannian in $\widehat{\mathfrak{gl}}_\infty$ into differential operators that act as infinitesimal symmetries of the KP hierarchy. 

There are several examples that are important in this paper. First of all, we have:
$$
\alpha_n \leftrightarrow \widehat J_n := \frac{\d}{\d t_n}, n>0; \qquad
\alpha_n \leftrightarrow \widehat J_n := - n t_{-n}, n<0,
$$
where the operators $\widehat J_n$ are defined on $\mathbb{C}[[\mathbf{t}]]$.
The energy operator $E\colon V_0\to V_0$ defined as $E\colon v_\lambda\mapsto |\lambda|v_\lambda$ corresponds to the operator $\widehat L_0\colon \mathbb{C}[[\mathbf{t}]]\to \mathbb{C}[[\mathbf{t}]]$ defined as
$$
\widehat L_0 := \frac{1}{2}\sum_{i+j=0} \normordboson\widehat J_i \widehat J_j \normordboson,
$$
where the normal ordering denoted by $\normordboson\cdots \normordboson$ put all operators $\widehat{J}_k$ with positive $k$ to the right of all $\widehat{J}_k$ with negative $k$. 
The Casimir operator $\tilde \cE_0 (z)\colon V_0\to V_0$ acts as follows:
\begin{equation}
\tilde\cE_0(u) v_\lambda = \sum_{r=0}^\infty \frac{u^r}{r!} \sum_{i=1}^{\ell(\lambda)} \left[ (\lambda_i-i+\frac 12)^r - (-i+\frac 12)^r \right] v_\lambda.
\end{equation}
Using the auxiliary functions $\zeta(u)=e^{u/2}-e^{-u/2}$, we can present the corresponding differential operator on $\mathbb{C}[[\mathbf{t}]]$ as
\begin{equation}\label{gfCasim}
\frac{1}{\zeta(u)} \sum_{n=1}\frac{1}{n!} \sum_{\substack{\vec{k} \in (\Z^\times)^n \\ k_1+\cdots+k_n=0}} \prod_{i=1}^{n} \frac{\zeta(k_iu)}{k_i} 
\normordboson \widehat J_{k_1}\cdots \widehat J_{k_n} \normordboson
\end{equation}
(see~\cite{Rossi2008,Alexandrov2011,ShadrinSpitzZvonkine2012}). 

\subsection{Kac-Schwarz operators}

A convenient way to describe infinitesimal symmetries of the KP hierarchy is to work with the operators from the algebra $w_{1+\infty}$ (the algebra of differential operators in one variable that describes infinitesimal diffeomorphisms of the circle) acting on the basis vectors $\Phi_i$, $i=1,2,\dots$. 

Let us denote by $\left<\Phi\right>$ the point of the Sato Grassmannian, defined by the set of the basis vectos $\left<\Phi_1,\Phi_2,\Phi_3,\dots\right>$. We call an operator $a \in w_{1+\infty}$ the Kac--Schwarz (KS) operator for the tau-function $\tau$ if for the corresponding point of the Sato Grassmannian we have the stability condition
\begin{equation}\label{KScond}
a\left<\Phi\right> \subset \left<\Phi\right>.
\end{equation}

For the trivial tau-function $\tau_\emptyset:=1$ with the basis vectors
$\Phi^\emptyset_k=x^{1-k}$, $k\geq 1$,
we have two obvious KS operators
\begin{align}\label{trivKS}
a_\emptyset & :=-x\frac{\d}{\d x},
\\ \notag
b_\emptyset & :=x^{-1}.
\end{align}
These operators satisfy the commutation relation
\begin{equation}\label{trivcomr}
\left[a_\emptyset,b_\emptyset\right]=b_\emptyset.
\end{equation}
The KS operators (\ref{trivKS}) act on the basis vectors as follows:
\begin{align}\label{trivac}
a_{\emptyset}\,\Phi_k^{\emptyset}(x) & =(k-1)\Phi_k^{\emptyset}(x),\\ \notag
b_{\emptyset}\,\Phi_k^{\emptyset}(x) &= \Phi_{k+1}^{\emptyset}(x).
\end{align}

Consider the tau-function 
\begin{equation}\label{simp}
\tau_\bullet({\bf t}; {\bf \tilde{t}})=e^{\sum_{k=1}^\infty k t_k \tilde{t}_k}=\sum_{\lambda} s_\lambda({\bf t}) s_\lambda({\bf \tilde{t}}), 
\end{equation}
where $s_\lambda(t)$ are the Schur functions. From the point of view of the KP hierarchy this tau-function corresponds to the basis vectors
\begin{equation}\label{bulbv}
\Phi^\bullet_k(x)= e^{\sum_{j=1}^\infty\tilde{t}_j x^j} x^{1-k}
\end{equation}
and the KS operators can be obtained from \eqref{trivKS} by conjugation:
\begin{align}\label{KSbul}
a_\bullet & :=e^{\sum_{j=1}^\infty\tilde{t}_j x^j}a_\emptyset\,e^{-\sum_{j=1}^\infty\tilde{t}_j x^j}=\sum_{k=1}^\infty k \tilde{t}_k x^ k-x\frac{\d}{\d x},
\\ \notag
b_\bullet & :=e^{\sum_{j=1}^\infty\tilde{t}_j x^j}b_\emptyset\,e^{-\sum_{j=1}^\infty\tilde{t}_j x^j}=x^{-1}.
\end{align}
In this case the commutation relation and action of the KS operators on the basis vectors coincide with the ones given by Equations \eqref{trivcomr} and \eqref{trivac}.

Basis vectors for the points of the Sato Grassmannian, corresponding to the double Hurwitz numbers, can be obtained from (\ref{bulbv}) by an action of the operators which are formal series in $x\frac{\p}{\p x}$. Corresponding Kac-Schwarz operators (and, in particular, the quantum spectral curve operator) can be obtained from the operators (\ref{KSbul}) by a conjugation and also satisfy relations  \eqref{trivcomr} and \eqref{trivac}. 

\begin{remark}
Let us stress that the algebra of the Kac--Schwarz operators for the trivial tau-function is generated not by the operators $a_\emptyset$ and $b_\emptyset$, but instead by the operators $b_\emptyset^{-1} a_\emptyset$ and $b_\emptyset$. Of course, this is also true for the corresponding Kac--Schwarz operators for all tau-functions, which can be obtained from the trivial one by a simple conjugation, in particular for the generating functions of the Hurwitz numbers (see Remark \ref{caj}).
\end{remark}



\section{Symmetric polynomials of Jucys elements}

\label{sec:Jucys}

In this section we briefly recall some relations between different bases of the algebra of symmetric polynomials and the Jucys correspondence. 

\subsection{Symmetric polynomials}
We consider the elementary symmetric polynomials $\sigma_b$, the complete homogeneous polynomials $h_b$, and the power sums $p_b$:
\begin{align*}
\sigma_b(x_1, \dots, x_n) & := \sum_{1 \leq i_1 < \cdots < i_b \leq n} x_{i_1}\cdots x_{i_b}, \\
h_b(x_1, \dots, x_n) & := \sum_{1 \leq \lambda_1 \leq \dots \leq \lambda_b \leq n} \!\!\!\! x_{\lambda_1}\cdots x_{\lambda_b}, \\
p_b(x_1, \dots, x_n) & := \sum_{1\leq i \leq n} x_i^b.
\end{align*}
The polynomials $\sigma_b$ and $h_b$ have the following generating series:
\begin{align}\label{gen}
\prod_{i=1}^n (1 + x_i t) & = \sum_{b=0}^{\infty} \sigma_b(x_1, \dots, x_n) t^b,
\\ \notag
 \prod_{i=1}^n \frac{1}{(1 - x_i t)} & = \sum_{b=0}^{\infty} h_b(x_1, \dots, x_n) t^b.
\end{align}

The Newton identities describe relations between the power sums $p_b$ and bases $\sigma_b$ and $h_b$:
\begin{align}\label{hsigmacoeff}
\sigma_b & = [z^b]. \exp \left( - \sum_{i \geq 1 } \frac{p_i}{i} (-z)^i \right),
\\ \notag
h_b & = [z^b]. \exp \left(\sum_{i \geq 1 } \frac{p_i}{i} z^i \right).
\end{align}
We also have the following relations between $\sigma_b$ and $h_b$:
\begin{align}\label{hcomesigma}
h_b  & =\sum_{k=1}^b (-1)^{k+b} \sum_{\substack{\vec{\alpha} \in (\N^{\times})^k \\ | \alpha | = b}}  \sigma_{\alpha_1}\cdots\sigma_{\alpha_k},
\\ \notag 
\sigma_b  & =\sum_{k=1}^b (-1)^{k+b} \sum_{\substack{\vec{\alpha} \in (\N^{\times})^k \\ | \alpha | = b}}  h_{\alpha_1}\cdots h_{\alpha_k}.
\end{align}

\subsection{The Jucys correspondence} Let $\alpha \in \mathfrak{S}_n /\!\!\sim$ be a conjugacy class of the symmetric group $\mathfrak{S}_n$ or, equivalently, a partition of $n$. We denote the  number of cycles of $\alpha$ by $\ell(\alpha)$.
We denote the formal sum of all permutations with cycle type $\alpha$ as
$C_{\alpha} := \sum_{g \in \alpha} g$. Note that $C_\alpha$ belongs to the center of the group algebra of $\mathfrak{S}_n$, that is, $C_\alpha\in \mathcal{Z}( \Q(\mathfrak{S}_n))$ for any $\alpha$. 
The elements $C_{\alpha}$ span $\mathcal{Z}( \Q(\mathfrak{S}_n))$.

We consider the Jucys-Murphy elements $\J_k\in \Q(\mathfrak{S}_n)$, $k=2,\dots,n$, defined as
\begin{equation}
\J_k := (1 \; k) + (2 \; k) + \dots + (k-1 \; k).
\end{equation}
They generate a maximal commutative subalgebra of $\Q(\mathfrak{S}_n)$ called Gelfand-Tsetlin algebra. 

The Jucys-Murphy elements are linked to the center of the group algebra through symmetric polynomials. 

\begin{lemma}[Jucys Correspondence \cite{Jucys1974}]\label{LEMMA1}\label{lem:JucysCorrespondence1} For $b=0,\dots,n-1$ we have:
	\begin{equation}\label{succoso}
	\sigma_b(\J_2, \dots, \J_n)  = \sum_{\substack{\alpha \in \mathfrak{S}_n/\sim \\ \ell(\alpha) = n - b}} C_{\alpha} .
	\end{equation}
\end{lemma}

This lemma together with the result of Farahat and Higman \cite{FarahatHigman1959} implies that symmetric polynomials in the Jucys-Murphy elements generate the center of the group algebra. 

Using Equation~\eqref{hcomesigma}, we obtain the following expression for the homogeneous complete polynomials of Jucys-Murphy elements:
\begin{lemma}
	\label{lem:JucysCorrespondence2}
For $b=0,\dots,n-1$ we have:
	\begin{equation}\label{hjuicy}
	h_b(\J_2, \dots, \J_n) 
	= \sum_{k=1}^b (-1)^{k+b} \sum_{\substack{\vec{\alpha} \in (\mathfrak{S}_n/\sim)^k \\ \sum \ell(\alpha_i) = kn - b}}  \prod_{i=1}^k  C_{\alpha_i} 
	\end{equation}
\end{lemma}


We denote $h_b(\J_2, \dots, \J_n)$ by $W^{\mathfrak{S}_n}_{b}$. Let $\mathscr{C}_{m}$ be the $m$-th Catalan number.  Let us list the first few examples of $W^{\mathfrak{S}_n}_{b}$:
\begin{align*}
W^{\mathfrak{S}_n}_{0} & =1;\\
W^{\mathfrak{S}_n}_{1} & = \mathscr{C}_1C_{(2^11^{n-2})}\ (= \sigma_1(\vec{\J}) = p_1(\vec{\J}) = \text{sum of all transpositions}); \\
W^{\mathfrak{S}_n}_{2} & = \mathscr{C}_2 C_{(3^11^{n-3})} + \mathscr{C}_1^2 C_{(2^21^{n-4})} + \frac{n(n-1)}{2}C_{(1^n)}; \\
W^{\mathfrak{S}_n}_{3} & = \mathscr{C}_3 C_{(4^11^{n-4})} +  \mathscr{C}_2\mathscr{C}_1 C_{(3^12^11^{n-5})} + \mathscr{C}_1^3 C_{(2^31^{n-6})} 
\\ & \phantom{ = }
+ \left(\frac{(n+1)(n+2)}{2}-5 \right) C_{(2^11^{n-2})}.
\end{align*}
Of course, each summand appears if and only if $n$ is big enough to allow the corresponding cycle type.


\section{Ramifications of Hurwitz theory}

\label{sec:HurwitzTheory}

In this Section we define the basic objects of study in this paper --- different variations of the classical Hurwitz numbers, whose definition utilizes symmetric functions of Jucys-Murphy elements. We describe a class of problems and their geometric interpretations. 

\subsection{General setup} The general setup is the following. We consider the coefficient of $C_{(1^n)}$ in the product $C_\mu C_\nu B$ for some $B \in \Q(\mathfrak{S}_n)$:
\begin{equation}\label{count}
h^{\bullet}_{\mu,\nu,B} := \frac{1}{n!} [C_{(1^n)}] C_{\mu}C_\nu B
\end{equation}
If  $B \in \mathcal{Z}(\Q(\mathfrak{S}_n))$, then its action in the left regular representation is given by the diagonal matrix 
$\pv(B)$, whose action in the irreducible representation $\lambda$ is multiplication by the eigenvalue $\pv_\lambda(B)$. The elements $C_{\alpha}$ lie in the center and their eigenvalues are given by $\pv_\lambda(C_\alpha):=|C_\alpha|\chi_\lambda(\alpha)/\dim \lambda$, where $\dim \lambda$ and $\chi_\lambda$ are the dimension and the character of the representation $\lambda$, and $|C_\alpha|$ is the number of permutations of the cycle type $\alpha$. This implies that
\begin{align} \label{eq:characterformula}
h^{\bullet}_{\mu,\nu,B}  & = \frac{\Tr\left(\pv(C_\mu)\pv(C_\nu)\pv(B)\right)}{(n!)^2} 
\\ \notag
& = \sum_{\lambda \vdash n} \left(\frac{\dim \lambda }{n!}\right)^2 \pv_\lambda(C_\mu)\pv_\lambda(C_\nu)\pv_\lambda(B)
\\ \notag
& =
 \frac{1}{Z_{\mu}Z_{\nu}}\sum_{\lambda \vdash n} \chi_{\lambda}(\mu)\chi_{\lambda}(\nu)\pv_\lambda(B),
\end{align}
where $Z_{\mu} = \prod \mu_i \prod_{i=1}^{n} (j_i)!$ for $\mu=(1^{j_1}2^{j_2}\cdots n^{j`_n})=(\mu_1\geq\cdots\geq \mu_{\ell(\mu)})$.

Let us discuss some examples. One can observe that
 \[
 \pv_\lambda(C_2)=\frac{1}{2}\sum_{i=1}^{\ell(\lambda)} (\lambda_i -i + \frac{1}{2})^2 - (-i+\frac{1}{2})^2.
 \] 
 The Hurwitz number $h^{\bullet}_{\mu,\nu,B}$ for $B=C_2^{2g-2+\ell(\mu)+\ell(\nu)}$ is the standard double Hurwitz number for possibly disconnected surfaces of genus $g$~\cite{Okounkov2000}. Consider an element $\overline{C}_r$ such that  
\[
\pv_\lambda(\overline{C}_r)=\frac{1}{r!}\sum_{i=1}^{\ell(\lambda)} (\lambda_i -i + \frac{1}{2})^{r} - (-i+\frac{1}{2})^r. 
\]
It is the so-called completed $r$-cycle~\cite{KerovOlshanski1994} (in some normalization), and the Hurwitz number $h^{\bullet}_{\mu,\nu,B}$ for $B=\overline{C}_r^m$, $m(r-1)=2g-2+\ell(\mu)+\ell(\nu)$,
is the double Hurwitz number with completed $r$-cycles for possibly disconnected surfaces of genus $g$~\cite{ShadrinSpitzZvonkine2012}.

In some cases, one can consider the enumeration of coverings up to automorphisms that fix the preimages of two special points (say, $0$ and $\infty$ in $\CP^1$) pointwise. In this case, we use the following formula instead of the one given by Equation~\eqref{eq:characterformula}:
\[
\frac{1}{\prod_{i=1}^{\ell(\mu)} \mu_i \prod_{i=1}^{\ell(\nu)} \nu_i }\sum_{ \lambda \vdash n} \chi_{\lambda}(\mu)\chi_{\lambda}(\nu)\pv_\lambda(B).
\]

\subsection{Basic definitions}
Let $\rho$ be a standard Young tableau of a Young diagram $\lambda\vdash n$. We denote by $i_k$ and $j_k$ the column and the row indices of the box labeled by $k$. By 
$$\mathsf{cr}^{\rho} := (i_1-j_1, i_2-j_2,\dots,i_n-j_n)$$ 
we denote the {content vector} of the tableau. Jucys~\cite{Jucys1974} proves that 
\begin{equation}\label{eq:JucysProperValue}
\pv_\lambda (B(\J_2, \dots, \J_n)) = B(\mathsf{cr}^{\rho}_2,\dots,\mathsf{cr}^{\rho}_n)
\end{equation}
for any symmetric polynomial $B$ in $n-1$ variables and any choice of $\rho$. Since it does not depend on $\rho$, we can always use some standard choice of the Young tableau, for instance, filling the diagram from left to right, and denote by $\mathsf{cr}^{\lambda}$ the content vector for this choice. This implies the following:
\begin{lemma}\label{zeta}
	If $B=B(\J_2,\dots,\J_n)$ is a symmetric polynomial in the Jucys elements, then 
	$$  h^{\bullet}_{\mu,\nu,B}=\frac{1}{Z_{\mu}Z_{\nu}}\sum_{ \lambda \vdash n} \chi_{\lambda}(\mu)\chi_{\lambda}(\nu)B(\mathsf{cr}^{\lambda}_2,\dots,\mathsf{cr}^\lambda_n).$$
\end{lemma}

\begin{definition}\label{def:HurwitzProblem} A disconnected double Hurwitz problem is the following set of data: genus $g$, degree $n$, two partitions $\mu,\nu\vdash n$, and a vector $\vec{\mathcal{P}}=(\mathcal{P}_1,\dots,\mathcal{P}_m)$, $m\geq 1$, where each $\mathcal{P}_i$ is a central element of $\Q(\mathfrak{S}_n)$. We assign to each $\mathcal{P}_i$ a number $b_i$ and we require the Riemann-Hurwitz equation to hold: $\sum_{i=1}^m b_i = 2g-2+\ell(\mu)+\ell(\nu)$. The associated Hurwitz number is then $h^\bullet_{\mu,\nu,B}$ for $B:=\prod_{i=1}^m \mathcal{P}_i$, and it can be expressed as 
	\begin{equation}\label{sol}
	h^{\bullet}_{\mu, \nu,B} =
	\frac{1}{n!} [C_{(1^n)}] C_{\mu}C_\nu \prod_{i=1}^m \mathcal{P}_i
	\end{equation}
We call elements $\mathcal{P}_i$ blocks and the vector $\vec{\mathcal{P}}$ the vector of blocks. 

Here are some possible blocks (that is, the possible values of $\mathcal{P}_i$, $i=1,\dots,m$), which are arguably most important for applications:
\begin{align*}
B^{<}_{b}  := \sigma_{b}(\J_2, \dots, \J_n); \
 B^{\leq}_b & := h_{b}(\J_2, \dots, \J_n); \
B^{\times}_{b}  := p_b(\J_2, \dots, \J_n); \\
B^{|}_{b}  := \sum_{\substack{\alpha \in (\mathfrak{S}_n/\sim) \\ \ell(\alpha)=n-b }} C_{\alpha}; \ {\ }
& B^{||}_{b}   := \sum_{k=1}^b (-1)^{k+b} \sum_{\substack{\vec{\alpha} \in (\mathfrak{S}_n/\sim)^k \\ \sum \ell(\alpha_i) = kn - b}} \prod_{i=1}^k  C_{\alpha_i} 
\end{align*}
In all these cases $b_i:=b$. 
\end{definition}

In each of this cases we can describe the geometry of the covering that realizes the monodromy of the block. The descriptions follow directly from the definition of the Jucys-Murphy elements $\J_k$, $k=2,\dots,n$, and the central elements $C_\lambda$, $\lambda\vdash n$. 

\begin{lemma}\label{blocks} The geometric interpretation of the possible blocks is the following:
\begin{itemize}
\item[$B^{<}_{b}$] \emph{[Strictly Monotone]} We have $b$ simple ramifications, whose mono\-dromies are given by the transpositions $(x_i \; y_i)$, $x_i<y_i$, $i=1,\dots,b$, with the extra condition $y_i < y_{i+1}$.

\item[$B^{\leq}_b$] \emph{[Monotone]} We have $b$ simple ramifications, whose monodromies are given by the transpositions $(x_i \; y_i)$, $x_i<y_i$, $i=1,\dots,b$, with the extra condition $y_i \leq y_{i+1}$.

\item[$B^{\times}_b$] \emph{[Atlantes]} We have $b$ simple ramifications, whose monodromies are given by the transpositions $(x_i \; y)$, $i=1,\dots,b$. Here $y$ is an arbitrary number from $2$ to $n$, which is not fixed in advance, but is the same for all transpositions. 

\item[$B^{|}_b$] \emph{[Free Single]} We have one ramification, whose monodromy has no restrictions except for the Euler characteristic of the preimage of the corresponding disk, that is, the monodromy given by a cycle type $\mu$ with $\ell(\mu)=n-b$.

\item[$B^{||}_b$] \emph{[Free Group]} We have an arbitrary number $k$ of ramifications, $1\leq k\leq b$ with no restrictions on the monodromy except for the restriction on the Euler characteristic: the total number of zeros of the differential of the corresponding covering should be equal to $b$.  The coverings are counted with an extra sign $(-1)^{k + b}$.
\end{itemize}
\end{lemma}

The Jucys correspondence given by Lemmas~\ref{lem:JucysCorrespondence1} and~\ref{lem:JucysCorrespondence2} implies the following equalities:
\begin{proposition}\label{jucysonblocks} We have
	$B^{<}_b= B^{|}_{b}$ and $B^{\leq}_{b} = B^{||}_{b}$.
\end{proposition}

\subsection{Examples} Here we survey some examples of disconnected double Hurwitz problems in the sense of Definition~\ref{def:HurwitzProblem} known in the literature.

\subsubsection{The Harnad-Orlov correspondence}
In \cite{HarnadOrlov2015} Harnad and Orlov prove that a family of 2D Toda tau-functions of hypergeometric type have two different geometric interpretations involving double Hurwitz problems. Their Theorem 2.1 expresses these tau-functions in terms of some Hurwitz numbers of some special type and their Theorem 2.2 deals with enumeration of paths in Cayley graphs. We review these two theorems and show that Jucys correspondence implies their equivalence. 

The hypergeometric function $\tau_{(q,w,z)}(\mathbf{t},\mathbf{\tilde{t}})$ is defined as
\begin{equation}\label{HarnadO}
\tau_{(q,w,z)}(\mathbf{t},\mathbf{\tilde{t}}) := 
\sum_{n=0}q^n \sum_{\lambda\vdash n} \prod_{j=1}^n \frac{\prod_{a=1}^l (1 + \mathsf{cr}^{\lambda}_j w_a)}{\prod_{b=1}^m (1 - \mathsf{cr}^{\lambda}_j z_b)} s_\lambda(\mathbf{t}) s_\lambda(\mathbf{\tilde{t}})
\end{equation}
Here $w=(w_1,\dots,w_l)$ and $z=(z_1,\dots,z_m)$ are the parameters of the tau-function, and their number ($l$ and $m$ respectively) is arbitrary, not necessarily finite. For particular values of these parameters the  hypergeometric tau-functions represent all generating functions of the Hurwitz numbers considered below. Using the generating functions of $\sigma_b$ and $h_b$ (see Equation \eqref{gen}) we rewrite (\ref{HarnadO}) as
	$$\sum_{n=0}q^n \sum_{\lambda\vdash n} \sum_{\substack{c \in \N^l \\ d \in \N^m}} \prod_{a=1}^l w_a^{c_a} \sigma_{c_a}(\mathsf{cr}^{\lambda})   \prod_{b=1}^m z_b^{d_b} h_{d_b}(\mathsf{cr}^{\lambda})
	s_\lambda(\mathbf{t}) s_\lambda(\mathbf{\tilde{t}})
	$$
Since $s_\lambda(\mathbf{t}) = \sum_{\mu\vdash |\lambda|} \chi_\lambda(\mu) p_\mu(\mathbf{t}) / Z_\mu$, the coefficient of 	
$$
q^n \prod_{a=1}^l \prod_{b=1}^m  w_a^{c_a}z_b^{d_b} p_\mu(\mathbf{t}) p_\nu(\mathbf{\tilde t})
$$
in this expression is equal to  
$$
\sum_{\lambda\vdash n} \frac{\chi_\lambda(\mu)  \chi_\lambda(\nu)}{Z_\mu Z_\nu}\prod_{a=1}^l \sigma_{c_a}(\mathsf{cr}^{\lambda}) \prod_{b=1}^m  h_{d_b}(\mathsf{cr}^{\lambda})  
$$
Lemma~\ref{zeta}, Lemma~\ref{blocks}, and Jucys statement about eigenvalues~\eqref{eq:JucysProperValue} imply that this coefficient is equal to $h^{\bullet}_{\mu, \nu,B}$ for the vector of blocks given by 
$$
\vec{\mathcal{P}}:= \left(B^{<}_{c_1},\dots,B^{<}_{c_l},B^{\leq}_{d_1},\dots,B^{\leq}_{d_m}\right)
$$
This is the way Harnad and Orlov prove \cite[Theorem 2.2]{HarnadOrlov2015}. Now Proposition~\ref{jucysonblocks} implies a different interpretation of the same Hurwitz number, namely, 
$$
\vec{\mathcal{P}} = \left(B^{|}_{c_1},\dots,B^{|}_{c_l},B^{||}_{d_1},\dots,B^{||}_{d_m}\right),
$$
which proves \cite[Theorem 2.1]{HarnadOrlov2015}.

\begin{remark}
	The monotone and strictly monotone blocks are expressed in \cite{HarnadOrlov2015} as counting paths in the Cayley graph of $\mathfrak{S}_n$. For convenience, we express them as a Hurwitz problem here. 
\end{remark}
\begin{remark}
	We also adjust a small inconsistency: observe that our weight in each summand of the free group block is $(-1)^{k+b}$, while in \cite{HarnadOrlov2015} it is $(-1)^{n + k + b}$.  
\end{remark}

\begin{remark}
	The solutions of the Hurwitz problem in genus zero with $\nu = (1^n)$ and a single block $B^{||}_{b, k}$ (which coincides with $B^{||}_{b}$ except that the number of groups $k$ is fixed and it is not weighted by sign) is known as Bousquet-M\'{e}lou Shaeffer numbers~\cite{BMS2000}, see also~\cite{KazarianLando2015}.
\end{remark}

\subsubsection{Enumeration of hypermaps}

The enumeration of hypermaps, or, more generally, of Gro\-thendieck's dessins d'enfants, is considered in many recent papers in slightly different formulations in relation to the Chekhov--Eynard--Orantin recursion, quantum curves, and KP/Toda integrability. An incomplete list of recent references includes~\cite{dunin2014combinatorics,KazarianZograf2015,DoManescu2014,giorgetti2015enumeration,Zograf2013,GouldenJackson2008,AlexandrovMironovMorozovNatanzon2014,AmbjornChekhov2014}.

Enumeration of hypermaps is equivalent to the standard weighted count of the coverings of degree $n$ of a sphere $\CP^1$ by a surface of genus $g$ (or, rather, a possibly disconnected surface of Euler characteristic $2-2g$) that have three ramification points, $0$, $1$, and $\infty$, such that
\begin{itemize}
	\item[--] The monodromy over $0$ has cycle type $\mu\vdash n$, which is a parameter of the enumeration problem.
	\item[--] The monodromy over $\infty$ has cycle type $(r^{n/r})$, $r$ is a parameter of the enumeration problem, and we assume that $r|n$. 
	\item[--] The monodromy over $1$ is an arbitrary one. Let us denote it by $\kappa\vdash n$. The only restriction that we have here is imposed by the Riemann-Hurwitz formula $2g-2+\ell(\mu)+n/r=n-\ell(\kappa)$.
\end{itemize}
In our terms, this enumeration problem can be reformulated as a Hurwitz number $h^{\bullet}_{\mu,\nu,B}$, where $\nu=(r^{|\mu|/r})$, the vector of blocks $\vec{\mathcal{P}}=(\mathcal{P}_1)$ has length 1, and 
\[
B=\mathcal{P}_1:=B^|_{2g-2+\ell(\mu)+|\mu|/r}.
\] 
Proposition~\ref{jucysonblocks} implies that 
\[
B=B^<_{2g-2+\ell(\mu)+|\mu|/r}.
\] 

The Hurwitz numbers for the data 
\[
(g,n,\mu,\nu=(r^{n/r}), \vec{\mathcal{P}}=(B^\leq_{2g-2+\ell(\mu)+n/r}))
\]
are called monotone orbifold Hurwitz numbers in~\cite{DoKarev2015} (orbifold here refers to the type of partition $\nu$), so it is natural to call Hurwitz numbers for the data 
\[
(g,n,\mu,\nu=(r^{n/r}), \vec{\mathcal{P}}=(B^<_{2g-2+\ell(\mu)+n/r}))
\]
strictly monotone orbifold Hurwitz numbers. Then the observation above can be reformulated as follows:

\begin{proposition}\label{prop:hypermaps-strictlymon}
	The enumeration of hypermaps is equivalent to the strictly monotone orbifold Hurwitz problem.
\end{proposition}


\begin{remark} This proposition also implicitly follows from the discussion in~\cite[Section 1]{AlexandrovMironovMorozovNatanzon2014}, in a different way.
\end{remark}



\section{Operators for $B^<_b$, $B^\leq_b$, and $B^\times_b$}
\label{sec:OperatorsForBlocks}

In this Section we derive the operators that represent the blocks $B^<_b$, $B^\leq_b$, and $B^\times_b$ in the semi-infinite wedge formalism and provide the corresponding differential operators.

\subsection{Derivation of operators} Recall that the Casimir operator $\tilde \cE_0 (z)$ on $V_0$
is a $\widehat{\mathfrak{gl}}_\infty$-operator (\ref{gfCasim}) that generates completed cycles. We would like to construct the same operators for the blocks $B^<_b$, $B^\leq_b$, and $B^\times_b$, that is, we are looking for the operators $\cD^{(p)}(z)$, $\cD^{(h)}(z)$, and $\cD^{(\sigma)}(z)$ defined on $V_0$ and acting on the basis vectors as follows:
\begin{align*}
\cD^{(p)}(z) v_\lambda & =  \sum_{k=1}\frac{z^k}{k!} p_k(\textsf{cr}^{\lambda})v_{\lambda}, &
\cD^{(h)}(z) v_\lambda & =  \sum_{k=0}z^{k} h_k(\textsf{cr}^{\lambda})v_{\lambda},\\
\cD^{(\sigma)}(z) v_\lambda & =  \sum_{k=0} z^{k} \sigma_k(\textsf{cr}^{\lambda})v_{\lambda}. 
\end{align*}
Since it is not important how we arrange the generating functions, we do it in the way that is most convenient for the proof below. 

\begin{remark} 
While $\cD^{(p)}(z)$ is an element of the $\widehat{\mathfrak{gl}}_\infty$ Lie algebra, operators $\cD^{(h)}(z)$ and $\cD^{(\sigma)}(z)$ belong to the corresponding group. From the Newton identities it follows that
\begin{equation*}
\cD^{(h)}(z)=\frac{1}{\cD^{(\sigma)}(-z) }.
\end{equation*}

\end{remark}

\begin{proposition} These operators, as the formal series in $z$, are given by the following formulas:
\begin{align*}
\cD^{(p)}(z) & =  \frac{\tilde \cE_0(z)}{\zeta(z)} - E, \\
\cD^{(h)}(z) & =  z^{\frac{\tilde \cE_0\left(z^2\frac{d}{dz}\right)}{\zeta\left(z^2\frac{d}{dz}\right)} - E} := \exp\left(\left[\frac{\tilde \cE_0\left(z^2\frac{d}{dz}\right)}{\zeta\left(z^2\frac{d}{dz}\right)} - E\right] \log z \right) ,\\
\cD^{(\sigma)}(z) & =  z^{-\frac{\tilde \cE_0\left(-z^2\frac{d}{dz}\right)}{\zeta\left(-z^2\frac{d}{dz}\right)} +E}
:= \exp\left(-\left[\frac{\tilde \cE_0\left(-z^2\frac{d}{dz}\right)}{\zeta\left(-z^2\frac{d}{dz}\right)} - E\right] \log z \right). 
\end{align*}
\end{proposition}

\begin{proof} The action of the power sums of Jucys elements was computed by Lascoux and Thibon in~\cite[Proposition 3.3]{LascouxThibon2001}. The formula for $\cD^{(p)}(z)$ is equivalent to their result. Note that the constant term of $\tilde \cE(z) / \zeta(z)$ is precisely $E$. The formulas for $\cD^{(h)}(z)$ and $\cD^{(\sigma)}(z)$ follow from the Newton identities~\eqref{hsigmacoeff}.
\end{proof}

\begin{remark} Since we know the differential operator (\ref{gfCasim}) that corresponds to $\tilde{\mathcal{E}}_0$, we immediately obtain the differential operators corresponding to $\mathcal{D}^{(p)}$, $\mathcal{D}^{(h)}$, and $\mathcal{D}^{(\sigma)}$.
\end{remark}

\begin{remark}
	Note that the formula for $\cD^{(\sigma)}(z)$ was already observed in~\cite[Section 3]{AlexandrovMironovMorozovNatanzon2014}.
\end{remark}

\begin{remark} The operators  $\cD^{(p)}(z_1)$, $\cD^{(h)}(z_2)$, and $\cD^{(\sigma)}(z_3)$ commute with each other for arbitrary values of $z_1$, $z_2$ and $z_3$.
\end{remark}

\subsection{Some examples} 
\label{subsec:SomeExamples}
In this Section we list some examples of particular Hurwitz problems whose generating functions are written as vacuum expectations in semi-infinite wedge formalism.  

\begin{example} 
Simple orbifold Hurwitz numbers:
	\begin{align*}
	\mathcal{Z}(\mathbf{p};\beta) = \cor{ \exp\left(\sum_{i=1}^{\infty} \frac{p_i\alpha_i}{i}\right)\exp\left(\hbar \cF_2\right) \exp\left(\frac{\alpha_{-r}}{\beta r}\right) },
	\end{align*}
	where $\cF_2= [z^1] \mathcal{D}^{(\sigma)}(z)$ is the second Casimir.
	Note that here we could use $\cD^{(h)}$ instead of $\cD^{(\sigma)}$ since their $[z^1]$ coefficients coincide.  
\end{example}

\begin{example} \label{example:one-parameter}
	A one-parameter deformation of simple Hurwitz numbers in the tau-function of double Hurwitz numbers:
\begin{equation*}
\mathcal{Z}(\mathbf{p},\beta)=\cor{\exp\left(\sum_{i=1}^\infty \frac {p_i\alpha_i}{i} \right) \exp \left(\beta\cF_2\right) \exp\left(\sum_{i=1}^\infty {\frac{c^{i-1}\alpha_{-i}}{\beta}}\right)}
\end{equation*}
\end{example}
 
\begin{example}\label{example:monotone}
	Monotone orbifold Hurwitz numbers:
	\begin{equation*}
	\mathcal{Z}(\mathbf{p};\beta) = \cor{\exp \left(\sum_{i=1}^{\infty} \frac{p_i\alpha_i}{i}  \right) 
		\mathcal{D}^{(h)}(\beta)
		\exp\left(\frac{\alpha_{-r}}{\beta r }\right)}
	\end{equation*}
\end{example}

\begin{example}\label{example:strictlymonotone}
	Strictly monotone orbifold Hurwitz numbers (or hypermaps, see Proposition~\ref{prop:hypermaps-strictlymon}):
	\begin{equation*}
	\mathcal{Z}(\mathbf{p};\beta) = \cor{\exp \left(\sum_{i=1}^{\infty} \frac{p_i\alpha_i}{i}  \right) 
		\mathcal{D}^{(\sigma)}(\beta)
		\exp\left(\frac{\alpha_{-r}}{\beta r }\right)}
	\end{equation*}
\end{example}

\begin{example} \label{example:atlantes}
	Atlantes orbifold Hurwitz numbers:
	\begin{align*}
	\mathcal{Z}(\mathbf{p};\beta) & =\cor{\exp \left(\sum_{i=1}^{\infty} \frac{p_i\alpha_i}{i}  \right) 
		\exp(r! [z^r]	\mathcal{D}^{(p)}(z \beta))
		\exp\left(\frac{\alpha_{-q}}{\beta q}\right)}
	\end{align*}
\end{example}



\section{Monotone Hurwitz numbers}

\label{sec:MonotoneHurwitz}

In this Section we discuss the monotone (orbifold) Hurwitz numbers (see Example~\ref{example:monotone} above) from different points of view. 

\subsection{HCIZ matrix integral and basis vectors}

According to~\cite{GouldenGuayNovak2014} the generating function of double monotone Hurwitz numbers is described by the Harish-Chandra--Itzykson--Zuber (HCIZ) tau-function. More precisely, let us introduce the tau-function
\begin{equation}\label{HCIZ}
\tau_{HCIZ}\left({\bf t}, {\bf\tilde{t}} , \alpha , N\right)=\sum_{\lambda} \alpha^{|\lambda|} s_\lambda({\bf t}) s_\lambda({\bf \tilde{t}}) \prod_{i}\frac{\Gamma(N-i+1)}{\Gamma(\lambda_i+N-i+1)},
\end{equation}
so that the HCIZ matrix integral is given by the Miwa parametrization $t_i=\Tr A^i$, $\tilde t_i = \Tr B^i$, $i=1,2,\dots,$ of this tau-function
\begin{equation}\label{HCIZmm}
\int dU e^{\alpha \Tr U A U^{\dagger}B}=\tau_{HCIZ}\left({\bf t}, {\bf\tilde{t}},\alpha,N\right).
\end{equation} 
Here we assume that the $N\times N$ matrices $A$ and $B$ are diagonal, and we normalize the Haar measure on the unitary group $U(N)$ in such a way that $\int dU=1$. Up to a factor that is not relevant for our computations, HCIZ integral describes a tau-function of the two-dimensional Toda lattice~\cite{MorozovMironovSemenoff1996,ZinnJustin2002}.

The generating function of the double monotone Hurwitz numbers is given by
\begin{align}\label{doublemontau}
\tau_{mm}\left({\bf t}, {\bf\tilde{t}} \right) & =\tau_{HCIZ}\left({\bf t}, {\bf\tilde{t}} , -\beta^{-1} ,-\beta^{-1} \right) \\ \notag 
&=\sum_{\lambda}  s_\lambda({\bf t}) s_\lambda({\bf \tilde{t}}) \prod_{i=1}^{l(\lambda)}\prod_{k=0}^{\lambda_i-1}\frac{1}{1+\beta(k+i-\lambda_i)},
\end{align}
or, in terms of the semi-infinite wedge product, by
\begin{align}
\tau_{mm}\left({\bf t}, {\bf\tilde{t}} \right) & =  \cor{\exp \left(\sum_{i=1}^{\infty} t_i\alpha_i  \right) 
		\mathcal{D}^{(h)}(\beta)
		\exp \left(\sum_{i=1}^{\infty} \tilde{t}_i\alpha_{-i}  \right)}.
\end{align}
%
%
%
\begin{proposition}
We can choose basis vectors of the KP hierarchy with respect to the set of times $\mathbf{t}$ in the following way:
\begin{equation}\label{opre}
\Phi_k^{mm}(x)=G_{mm}(k)\, e^{\sum_{m=1}^\infty   \tilde{t}_m x^ m } x^{1-k},
\end{equation}
where
\begin{equation}
G_{mm}(k)=\frac{\Gamma(1-k-\beta^{-1})(-\beta)^{1-k-D}}{\Gamma(D-\beta^{-1})}, \quad D=x\frac{\p}{\p x}.
\end{equation}
\end{proposition}

\begin{remark}
To specify the asymptotic we use the operator identity
\begin{equation}
f\left(\frac{\p}{\p a}\right) e^{a}=e^{a} f\left(\frac{\p}{\p a}+1\right)
\end{equation}
valid for arbitrary function $f$, so that for the leading coefficient of the series (\ref{opre}) we have
\begin{equation}
G_{mm}(k) x^{1-k}= x^{1-k}\frac{\Gamma(1-k-\beta^{-1})(-\beta)^{-D}}{\Gamma(1-k+D-\beta^{-1})} =x^{1-k}.
\end{equation}	
\end{remark}
We have
\begin{equation}
\Phi_k^{mm}(x)=\sum_{j=0}^\infty \frac{x^{j+1-k}}{\hbar^j j!} s_j({\bf t})\prod_{l=1}^j \frac{1}{1-\hbar(l-k)}.
\end{equation}
%



\subsection{Quantum curve from KS operators} \label{sec:QuantumKSmonotone}
We construct the KS operators by conjugation:
\begin{align}\label{qscmm}
a_{mm}=G_{mm}\, a_\bullet\,G_{mm}^{-1}
& =\sum_{k=1}^\infty k \tilde{t}_k x^k \frac{(-\beta)^{-k}\Gamma(D-\beta^{-1})}{\Gamma(D+k-\beta^{-1})} -D\\ \notag
& =\sum_{k=1}^\infty k \tilde{t}_k x^k\prod_{j=0}^{k-1}\frac{1}{1-\beta(D+j)} -D.
\end{align}
It follows from Equation~\eqref{trivac} that this operator annihilates the first basis vector. This implies that the wave function given by
\begin{align*}
\Psi^{mm}(x,\hbar)
& =\Phi_1^{mm}(x)\big|_{\tilde{t}_k\mapsto\tilde{t}_k/\hbar}
\\
& =\frac{\Gamma(-\hbar^{-1})(-\hbar)^{-D}}{\Gamma(D-\hbar^{-1})} e^{\frac{1}{\hbar}\sum_{k=1}^\infty   \tilde{t}_k x^k }
\end{align*}
is annihilated by the operator $A_{mm}$, where 
\begin{equation} \label{eq:FirstAmmForMonotone}
A_{mm}:=\sum_{k=1}^\infty k \tilde{t}_k x^k \prod_{j=0}^{k-1}\frac{1}{1-\hbar(D+j)} -\hbar D.
\end{equation} 
We call the operator $A_{mm}$ a general quantum curve. 

If $\tilde{t}_k=0$ for all $k>l$ with some finite $l$, then the quantum curve can be reduced to a polynomial one:
\begin{align*}
A_{mm}&=\sum_{k=1}^l k \tilde{t}_k\left( \prod_{j=1}^{k}\frac{1}{1-\hbar(D-j)} \right)x^k -\hbar D
\\
& =\left(\prod_{j=1}^{l}\frac{1}{1-\hbar(D-j)}\right)\tilde{A}_{mm},
\end{align*} 
where
\begin{equation*}
\tilde{A}_{mm}:=\sum_{k=1}^l k \tilde{t}_k x^k \prod_{j=1}^{l-k}\left(1-\hbar(D-j)\right)-\hbar D\prod_{j=1}^{l}\left(1-\hbar(D-j)\right)
\end{equation*}
also annihilates the wave function:
\begin{equation}
\tilde{A}_{mm}\Psi^{mm}(x,\hbar)=0.
\end{equation}
Introducing the operators $\hat{x}=x\cdot$,  $\hat{y}=-\hbar \frac{\p}{\p x}$, we obtain
\begin{equation}\label{qsccan}
\tilde{A}_{mm}=\sum_{k=1}^l k \tilde{t}_k \hat{x}^k \prod_{j=1}^{l-k}\left(1+\hat{x}\hat{y}+\hbar j\right)+\hat{x}\hat{y}\prod_{j=1}^{l}\left(1+\hat{x}\hat{y}+\hbar j\right).
\end{equation}
Further specializations of this formula imply the following proposition:

\begin{proposition} The quantum curve for the monotone $r$-orbifold Hurwitz numbers is equal to 
\begin{equation}
\hat{x}\left(\hat{x}^{r-1}+\prod_{j=1}^{r}\left(1+\hat{x}\hat{y}+\hbar (j-1)\right)\hat{y}\right).
\end{equation}
In particular, for $r=1$, it reduces to
\begin{equation}
\hat{x}(\hat{x}\hat{y}^2+\hat{y}+1).
\end{equation}
\end{proposition}

\begin{remark}
These expressions, up to a factor $\hat{x}$, coincide with the quantum curves obtained in~\cite{DoKarev2015,DoDyerMathews2014}.
\end{remark}

\begin{proof} Monotone $r$-orbifold Hurwitz numbers correspond to the specialization
\begin{equation}
\tilde{t}_k=\frac{\delta_{k,r}}{r}.
\end{equation}
In this case the quantum spectral curve (\ref{qsccan}) reduces to
\begin{align}
A&=\hat{x}^r+\hat{x}\hat{y}\prod_{j=1}^{r}\left(1+\hat{x}\hat{y}+\hbar j\right)
\\ \notag
& =\hat{x}\left(\hat{x}^{r-1}+\prod_{j=1}^{r}\left(1+\hat{x}\hat{y}+\hbar (j-1)\right)\hat{y}\right)
\end{align}
\end{proof}

\subsection{Linear equations for the tau-function} In this Section we derive some linear equations for the tau-function of double monotone Hurwitz numbers $\tau_{mm}(\mathbf{t},\tilde{\mathbf{t}})$.

Recall that the boson-fermion correspondence allows us to translate the operators in $w_{1+\infty}$ into the differential operators in the variables $\mathbf{t}$ in $\widehat{\mathfrak{gl}}_\infty$. The general formula reads:
\begin{equation}\label{obdifgen}
\left(xD \right)^m x^{k} \mapsto \widehat{Y}_{\left(xD \right)^m x^{k}}:=\Res_{x=0}\left(x^{-k}\normordboson \frac{(\widehat{J}(x)+\p_x)^{m}}{m+1} \widehat{J}(x)\normordboson \right) dx,
\end{equation}
where the operators $\left(xD\right)^m x^{k}$, $m\geq 0$, $k\in\mathbb{Z}$, span $w_{1+\infty}$, $D=x\frac{\p}{\p x}$. We refer to~\cite{AlexandrovTauFunctions2015} for a detailed exposition of this correspondence. 

\begin{remark} Note that the operator $\widehat{Y}_a$ is a finite-order differential operator if and only if $a\in w_{1+\infty}$ is a differential operator, that is a polynomial in $D$.
\end{remark}

\begin{proposition} The tau-function $\tau_{mm}(\mathbf{t},\tilde{\mathbf{t}})$ satisfies the following linear identities:
\begin{equation}
\widehat R_n \tau_{mm}(\mathbf{t},\tilde{\mathbf{t}}) = n\tilde t_n \tau_{mm}(\mathbf{t},\tilde{\mathbf{t}}),\quad n=1,2,\dots.
\end{equation}
where
\begin{equation}
\widehat R_n := \sum_{k=0}^n (-\beta)^{k} \sum_{1\leq i_1 < i_2 < \cdots < i_k \leq n} \widehat Y_{x^{-n}(D-i_1)\cdots (D-i_k)}.
\end{equation}
Moreover, this system of identities determines the tau-function uniquely up to a constant factor. 
\end{proposition}

\begin{proof} Consider the KS operator
\begin{align*}
 b_{mm}& :=G_{mm}\, b_\bullet\,G_{mm}^{-1} =x^{-1} \frac{(-\beta) \Gamma(D-\beta^{-1})}{\Gamma(D-1-\beta^{-1})}
=x^{-1}(1-\beta(D-1)).
\end{align*}
Using that $Dx^{-1}=x^{-1}(D-1)$, we have:
\begin{align*}
b^n_{mm}& =x^{-n}(1-\beta(D-1))(1-\beta(D-2))\cdots (1-\beta(D-n)).
\end{align*}
Hence we have that $\widehat R_n= \widehat{Y}_{b_{mm}^n}$, $n=1,2,\dots$.
Since the operators $b^n_{mm}$ are polynomial in $D$ and preserve $\{\Phi^{mm}\}$, 
the corresponding differential operators are finite degree operators in $\mathbf{t}$ that satisfy 
\begin{equation}\label{eq:constainsmm}
\widehat{Y}_{b^n_{mm}} \tau_{mm} = c_n(\tilde{\mathbf{t}}) \tau_{mm}.
\end{equation}
We have to determine the coefficients $c_n(\tilde{\mathbf{t}})$, $n=1,2,\dots$. 

Note that Equation~\eqref{eq:constainsmm} is obtained by conjugation with $\mathcal{D}^{(\sigma)}$, where $\mathcal{D}^{(\sigma)}$ is now considered also as a differential operator in $\mathbf{t}$, of the following equation for $\tau_\bullet(\mathbf{t},\tilde{\mathbf{t}})$:
\begin{equation}\label{eq:constrainsbullet1}
\widehat{Y}_{(b_\bullet)^n} \tau_\bullet = c_n(\tilde{\mathbf{t}}) \tau_\bullet, \quad, n=1,2,\dots.
\end{equation}
The last equation can be rewritten as  
\begin{equation}\label{eq:constrainsbullet2}
\frac{\d}{\d t_n} \tau_\bullet = n \tilde t_n \tau_\bullet, \quad, n=1,2,\dots.
\end{equation}
Thus we see that $c_n(\tilde{\mathbf{t}_n})=n\tilde t_n$, and since Equations~\eqref{eq:constrainsbullet2} determine the tau-function $\tau_\bullet$ up to a constant factor, the same is true for Equations~\eqref{eq:constainsmm} and the tau-function $\tau_{mm}$.
\end{proof}

\begin{remark} By construction, the operators $\widehat{R}_n$, $n\geq 1$, commute.
\end{remark}

\begin{example} \label{example:R1R2} Let us list the first two operators, $\widehat{R}_1$ and $\widehat{R}_2$. We have:
\begin{align*}
\widehat{R}_1 & = \frac{\p}{\p t_1}-\beta \widehat{L}_{1}, \\ 
\widehat{R}_2 & = \frac{\p}{\p t_2}-2\beta\widehat{L}_2+\beta^2 \widehat{M}_2,
\end{align*}
where
\begin{align*}
\widehat{L}_m& =\widehat{Y}_{x^{-m}\left(D-\frac{m+1}{2}\right)} =
\frac{1}{2} \sum_{a+b=m} \normordboson \widehat{J}_a \widehat{J}_b \normordboson,
\\
\widehat{M}_m& =\widehat{Y}_{x^{-m}\left(D^2-(m+1) D +\frac{(1+m)(2+m)}{6}\right)} = \frac{1}{3} \sum_{a+b+c=m} \normordboson \widehat{J}_a \widehat{J}_b \widehat{J}_c \normordboson
\end{align*}
are some standard infinitesimal symmetries of KP, see e.g.~\cite{AlexandrovTauFunctions2015}.
\end{example}

\begin{remark}\label{caj} For the tau-function of the double monotone Hurwitz numbers all possible Kac--Schwarz operators that are polynomial in $D$ are given by the polynomials of $b_{mm}$. However, for particular specializations of the parameters $\tilde{t}_k$, some other polynomial Kac--Schwarz operators can appear. In particular, for the single monotone Hurwitz numbers (Example \ref{example:monotone} with $r=1$) we have the following Kac--Schwarz operator:
\begin{equation*}
c_{mm}:=a_{mm}^2-b_{mm}^{-1}a_{mm}+(1+\hbar^{-1})a_{mm}=z-\hbar D+\hbar^2D(D-1).
\end{equation*}
The corresponding equation for the tau-function is
\begin{equation*}
\widehat{Y}_{c_{mm}}\left[\tau_{mm}({\bf t},\tilde{\mathbf{t}})|_{\tilde{t}_{k}=\hbar^{-1}\delta_{k,1}, k\geq 1}\right]=0,
\end{equation*}
where
\begin{equation*}
\widehat{Y}_{c_{mm}}=\hbar^2\widehat{M}_0-\hbar\widehat{L}_0+t_1,
\end{equation*}
is equivalent to the cut-and-join equation of \cite{GouldenGuayNovak2011-2}. Similar  operators can be easily found for higher $r$.
\end{remark}

\subsection{ELSV-type formula} We denote by $h_{g,\mu}^{\le} $ the monotone Hurwitz numbers for the connected covering surface of genus $g$. The generating function for these numbers is the logarithm of the one we have in Example~\ref{example:monotone} for $r=1$. 
The following quasi-polynomiality property is proved in~\cite{GouldenGuayNovak2011-2,GouldenPol}:
\begin{equation}
	h_{g,\mu}^{\le} = \prod_{i=1}^n \binom{2\mu_i}{\mu_i} P^{\le}_{g,n}(\mu_1,\dots,\mu_{\ell(\mu)})
\end{equation} 
	for some polynomial $P^{\le}_{g,n}$. Based on this formula the authors conjectured that there should be an ELSV-type formula for these numbers.

The topological recursion for these numbers is proved in~\cite{DoDyerMathews2014}. They prove that the expansions of the correlation differentials of the curve $x=-(y+1)/y^2$  are given by 
\begin{equation}\label{eq:monotonespectralcurve}
		\omega^{\leq}_{g,n}(z_1, \dots, z_n) =  d_1 \cdots d_n \, \sum_{\vec\mu\in (\N^\times)^n} h^{ \le}_{g,\mu} \prod_{i=1}^{\ell(\mu)} x_i^{\mu_i}.
\end{equation}
\begin{remark}
Here by $\vec{\mu}$ we denote a vector, that is, we don't assume that $\mu_1\geq \cdots\geq \mu_n$. We denote by $\mu$ the partition of length $n$ whose parts are the ordered components of the vector $\vec{\mu}$.
\end{remark}

This is sufficient to prove the following:

\begin{proposition}\label{elsvmon}We have:
	\begin{align} \label{eq:statementMonELSV}
	h_{g,\mu}^{\le}&  =  \prod_{i=1}^{\ell(\mu)} \binom{2\mu_i}{\mu_i} 
	\int_{\overline{\mathcal{M}}_{g,\ell(\mu)}} 
	 e^{{\sum_{l=1} K_l \kappa_l}}
	\prod_{j=1}^{\ell(\mu)} \sum_{d_j \geq 0} \psi_j^{d_j} \frac{(2(\mu_j + d_j) - 1)!!}{(2 \mu_j - 1)!!}.
	\end{align}
Here the coefficients $K_i$, $i=1,2,\dots$, satisfy the following equation:
	\begin{equation}\label{condK}
	\exp \left(-\sum_{l=1}^\infty K_l U^l \right)  = \sum_{k=0}^{\infty} (2k + 1)!! U^k.
	\end{equation}
\end{proposition}

\begin{remark}The cohomological field theory in this formula is given by the class $\exp\left({\sum_{l=1} K_l \kappa_l}\right)$ . This type of cohomological field theories of rank 1 with a non-flat unit is considered in detail in~\cite{ManinZograf2000}.
\end{remark}

	
\begin{proof}[Proof]
		First note that $x(y)$ has a single critical point $y_{cr} = -2$ with the critical value  $x_{cr}:= x(y_{cr})=1/4$. The local coordinate $\zeta$ around $y_{cr}$ and its inverse read
		$$ \zeta := \sqrt{x - x_{crit}} = i\frac{y+2}{2y};  \qquad y = \frac{2i}{2\zeta - i}.$$
We expand $y$ near $\zeta = 0$:
\begin{equation} \label{expyzeta}
y(\zeta) = \sum_{k \geq 0} s_k \zeta^k, \qquad s_{k} = i^k (-2)^{k+1}, 
\end{equation}
in particular for odd coefficients we have $s_{2k+1} = 4i (-4)^k$.

The correlation differentials $\omega_{g,n}$ produced by the topological recursion can be expressed as sums over graphs (see~\cite{EynardIntersections1,EynardInvariantsModuli2,DOSS}).
In the case when the spectral curve has a single branch point Theorem 3.3. in \cite{EynardIntersections1} gives an explicit formula for the $\omega_{g,n}$'s.
Since the local coordinate $\zeta$ is in fact a global coordinate on the sphere, the Bergman kernel is equal to 
\[
\frac{d\zeta_1d\zeta_2}{(\zeta_1-\zeta_2)^2}.
\]
This means that the Bergman kernel has trivial regular part near the critical point, and the expression in term of stable graphs simplifies sensibly since only stable graphs with a single vertex appear. It can be written as 
\begin{align}\label{asgraphs}
& \omega_{g,n}(\vec{\zeta})= 
\left(-{2 s_1} \right)^{2-2g-n}    \sum_{\vec{d} \in \mathbb{N}^n} \prod_{i=1}^n (2 d_i+1)!! \frac{d\zeta_i}{\zeta_i^{2 d_i+2}} \times 
\\ \notag
& \int_{\overline{\mathcal{M}}_{g,n}} \prod_{j=1}^n \psi_j^{d_j} \, \sum_{m=0}^\infty \frac{1}{m!}  \sum_{\vec{\alpha} \in (\mathbb{N}^\times)^m} \prod_{k=1}^m\left( -(2 \alpha_k+1)!! \frac{s_{2\alpha_k+1}}{s_1} \right)
\kappa_{{\alpha_1}, \dots, {\alpha_m}}
%
\end{align}
(cf. \cite[Equation (3.53)]{DOSS}), where the coefficients $s_k$ are given by Equation~\eqref{expyzeta}. 


In order to rewrite Equation~\eqref{asgraphs} as an expansion in $x_1,\dots,x_n$ near $y=-1$, we observe that
	\begin{equation}\label{eq:monotoneleaves}
	\frac{(2a + 1)!! d\zeta}{ \zeta^{2a + 2}} = 2i(-4)^a\, d\sum_{l=0}^\infty \binom{2l}{l} x^l \cdot \frac{(2(l+a)-1)!!}{(2l-1)!!}.
	\end{equation}
	Indeed, this follows from equation
		\begin{equation*}
		\frac{(2a + 1)!! d\zeta}{ \zeta^{2a+ 2}} = d\left( -\frac{d}{\zeta d\zeta} \right)^{a} (-\zeta^{-1}), 
		\end{equation*}
		and expansion 
		$$
		-\zeta^{-1} = 2i \sum_{l= 0}^{\infty} \binom{2 l}{ l} x^{l}.
		$$

The multi-index kappa classes can be written as exponent of sum of single kappa classes:
\begin{equation}\label{expk}
\sum_{m=0}^\infty \frac{1}{m!}  \sum_{\vec{\alpha} \in(\mathbb{N}^\times)^m} \prod_{k=1}^m f(\alpha_k)\kappa_{{\alpha_1}, \dots, {\alpha_m} } = \exp \left(\sum_{l = 1}^\infty K_l \kappa_l \right),
\end{equation}
where the coefficients $K_l$ can be computed by the expansion
\begin{equation}
\exp \left(- \sum_{l = 1}^\infty K_l U^l \right) = 1 - \sum_{k = 1}^\infty f(k) U^k.
\end{equation}
This implies that 
\begin{equation}\label{eq:monotonekappas}
\sum_{m=0}^\infty \frac{1}{m!}  \sum_{\vec{\alpha} \in (\mathbb{N}^\times)^m} \prod_{k=1}^m -(2 \alpha_k+1)!! \kappa_{{\alpha_1}, \dots, {\alpha_m} } = \exp\left(\sum_{l=1}^\infty K_l \kappa^l \right),
\end{equation}
where $\exp\left(-\sum_{l=1}^\infty K_l U^l \right) = \sum_{k=0}^{\infty} (2k + 1)!! U^k$.

Finally, observe that if $\sum_{k=1}^n d_k + \sum_{k=1}^m \alpha_m = 3g-3+n$, then
\begin{equation}\label{eq:thecoefficient}
(-2s_1)^{2-2g-n}\prod_{k=1}^n \left((-4)^{d_k}\cdot (2i)\right) \cdot \prod_{k=1}^m \frac{s_{2\alpha_k+1}}{s_1} = 1.
\end{equation}

Now we are ready to complete the proof of the proposition. 
	Note that Equation~\eqref{eq:monotonespectralcurve} implies that 
	\begin{equation}
	\sum_{\mu \in (\N^\times)^n} h^{\le}_{g,\mu} x_1^{\mu_1}\cdots x_n^{\mu_n}= \int \cdots \int \omega^{\leq}_{g,n}
	\end{equation}
	On the other hand Equations~\eqref{asgraphs}, 
  \eqref{eq:monotoneleaves}, \eqref{eq:monotonekappas}, and~\eqref{eq:thecoefficient} imply that 
\begin{align*}
	 & \int \cdots \int \omega^{\leq}_{g,n} = \\ 
	 &  \sum_{\vec{\mu} \in (\N^\times)^n} \sum_{\vec{d} \in (\N)^n}
	\int_{\overline{\mathcal{M}}_{g,n}}  e^{{\sum_{l=1} K_l \kappa_l}}
	\prod_{j=1}^{n}  \psi_j^{d_j} \frac{(2(\mu_j + d_j) - 1)!!}{(2 \mu_j - 1)!!} \binom{2\mu_j}{\mu_j} x_j^{\mu_j}
\end{align*}
for $K_l$ given by Equation~\eqref{condK}, which is equivalent to Equation~\eqref{eq:statementMonELSV}.
\end{proof}

\begin{remark} After we shared this formula with colleagues, we learned from N.~Do that he and M.~Karev derived the same formula independently, using the geometric approach to topological recursion due to M.~Kazarian. 
\end{remark}


\section{Further examples of quantum curves}
\label{sec:FurtherExamples}

\subsection{Strictly monotone orbifold Hurwitz numbers} By Proposition~\ref{prop:hypermaps-strictlymon} strictly monotone orbifold Hurwitz problem is equivalent to the enumeration of hypermaps. Its tau-function is given in Example~\ref{example:strictlymonotone}.

By principal specialization of Schur functions near infinity, the corresponding wave function is equal to
\begin{align} \label{eq:wavestrmonotone}
\Psi(x^{-1},\hbar)
& =  \sum_{n=0}^\infty \frac{x^{-rn}}{n! \hbar^n r^n} \sum_{k=0}^{\infty} \sigma_k(cr^{(rn,0,\dots,0)}) \hbar^k \\ \notag
& = \sum_{n=0}^\infty \frac{x^{-rn}}{n! \hbar^n r^n} \prod_{j=1}^{rn-1} (1 + j\hbar).
\end{align}
In order to get a curve, consistent with results of~\cite{DoManescu2014,DOPS2014}, here we consider the wave function as a series in the variable $x^{-1}$ instead of $x$.

\begin{proposition} We have: 
	$$
	\left[\hat x^{\frac 1\hbar} (\hat y^r - \hat x \hat y + 1) \hat x^{-\frac 1\hbar}\right] \Psi(x^{-1},\hbar) = 0,$$
	where $\hat x = x\cdot$ and $\hat y = -\hbar \frac{\p}{\p x}$.
\end{proposition}

\begin{proof}
   Let $a_n$ be the $n$th summand in Equation~\eqref{eq:wavestrmonotone}. We have:
	\begin{equation*}
	\hbar r(n+1) a_{n+1} = x^{-r} \prod_{j=0}^{r-1}[1 + (nr + j)\hbar] a_n.
	\end{equation*}
	In terms of the operators this can be rewritten as
	\begin{equation*}
	- \hbar x\frac{\d}{\d x} a_{n+1} =  \left[ x^{-1} \left(1 - \hbar x\frac{\d}{\d x}\right)\right]^r x^r a_n.
	\end{equation*}
	Hence we obtain
	$$x^{\frac 1\hbar}\left[ \left(-\hbar \frac{\d}{\d x}\right)^r + \hbar x \frac{\d}{\d x} + 1\right] x^{-\frac 1\hbar}\, \Psi(x^{-1},\hbar) = 0.$$
\end{proof}

\begin{remark}
This quantum curve was earlier obtained in~\cite{DoManescu2014} using combinatorics of hypermaps and in~\cite{DOPS2014} using the loop equations for hypermaps. Comparison with this results also forces us to use the variable $x^{-1}$ instead of $x$.
\end{remark}

\begin{remark} Even though we presented here a purely combinatorial derivation of the quantum curve, it is worth mentioning that one can derive it for a more general double strictly monotone Hurwitz problem using the method of Section~\ref{sec:QuantumKSmonotone}. In this case the operator given by Equation~\eqref{eq:FirstAmmForMonotone} is replaced by 
\[
	\sum_{k=1}^\infty k \tilde{t}_k x^k \prod_{j=0}^{k-1}(1+\hbar(D+j)) -\hbar D,
\]
whose specialization for $\tilde{t}_k = \delta_{k,r}/r$ is equivalent to the operator above after the change of variable $x\mapsto x^{-1}$.
\end{remark}

\subsection{Blocks of atlantes} 
 We consider a Hurwitz theory given by a vector of blocks of atlantes of some fixed type, that is, the vector of blocks is equal to $\vec{\mathcal{P}}=(B_r^\times,\dots,B_r^\times)$ for some fixed $r\geq 1$. We also assume that $\nu=(1^{|\mu|})$, see Example~\ref{example:atlantes} for $q=1$. The corresponding wave function is equal to
	\begin{align}\label{eq:waveatlantes}
		\Psi(x,\hbar) &= \sum_{n=0}^\infty \frac{x^n}{n! \hbar^n} \exp( p_r(cr^{(n,0,\dots,0)}) \hbar^r ) 
		\\ \notag &
		= \sum_{n=0}^\infty \frac{x^n}{n! \hbar^n} \exp\left( \hbar^r \sum_{j=1}^{n-1} j^r \right) 
	\end{align}
\begin{proposition} We have: 
	$$[\hat{y} - \hat{x}e^{\hat{y}^r}] \Psi(x,\hbar)= 0,$$
	where
	$\hat x = x\cdot$ and $\hat{y} = \hbar x \frac{\d}{\d x}$ (it is more convenient to use the exponential coordinate in this case, cf.~\cite{MulaseShadrinSpitz2013}).
\end{proposition}

\begin{proof}
	 Let $a_n$ be the $n$th summand in Equation~\eqref{eq:waveatlantes}. We have:
	\begin{equation*}
	\hbar (n+1) a_{n+1} = x e^{\hbar^r n^r} a_n 
	\end{equation*}
	In terms of the operators this can be rewritten as
	$$\hbar x \frac{\d}{\d x}a_{n+1} = x e^{(\hbar x \frac{\d}{\d x})^r}a_n.$$
	Therefore,
	$$
	\left[\hbar x \frac{\d}{\d x} - x e^{(\hbar x \frac{\d}{\d x})^r} \right] \Psi(x,\hbar)=0.
	$$
\end{proof}

\begin{remark}
	This case is very interesting since we can say in advance that the logic outlined in Section~\ref{sec:outlineoflogic} fails. Indeed, the dequantization of this quantum curve consides with the dequantization of the quantum curve for the $r$-spin Hurwitz number
	$$
	\hat{y} - \hat{x}^{3/2} \exp\left(\frac{1}{r+1} \sum_{i=0}^r \hat{x}^{-1} \hat{y}^{i} \hat{x} \hat{y}^{r-i}  \right)
	$$
	proved in~\cite{MulaseShadrinSpitz2013}. Even though the spectral curve and the corresponding $r$-ELSV formula for the $r$-spin Hurwitz numbers are still conjectural, there is a very strong evidence for these conjectures to be true~\cite{ShadrinSpitzZvonkine2015}. From these conjectures we can conclude that the dequantization of  $\hat y - \hat x \exp(\hat y^r)$ can not be the spectral curve for the atlantes Hurwitz numbers, suitable for the construction of the topological recursion.
	
	Indeed, even though in genus zero atlantes Hurwitz numbers coincide with the $r$-spin Hurwitz numbers (and hence all data of the spectral curve must be the same), in higher genera this is no longer the case.  
\end{remark}

\subsection{Double Hurwitz numbers}

The partition function of the double Hurwitz numbers is
\begin{align}\label{doublehurt}
	\tau_{HH}({\bf t}, {\bf \tilde{t}})= \cor{ \exp\left(\sum_{i=1}^{\infty} t_i\alpha_i\right)\exp\left(\hbar \cF_2\right)\exp\left(\sum_{i=1}^{\infty} \tilde{t}_i\alpha_{-i}\right) }.
\end{align}
or
\begin{equation*}\label{doubletau}
	\tau_{HH}({\bf t}; {\bf \tilde{t}})=\sum_{\lambda} s_\lambda({\bf t}) s_\lambda({\bf \tilde{t}}) e^{\beta  \pv_\lambda(C_2)}.
\end{equation*}
Then, the basis vectors for this tau-function as a KP tau-function with respect to times $t_k$ is
\begin{equation*}\label{basisHH}
	\Phi_k^{HH}(x)=e^{\frac{\beta}{2}\left((D-\frac{1}{2})^2-(k-\frac{1}{2})^2\right)} e^{\sum_{j=1}^\infty\tilde{t}_j x^j} x^{1-k}.
\end{equation*}
The wave function is given, as usual, by a rescaling of $\Phi_1^{HH}(x)$:
\[
\Psi(x,\hbar):= \Phi_1^{HH}(x)|_{\tilde{t}_{k}\mapsto \hbar^{-1}\delta_{k,1}, k\geq 1}
\]

\begin{proposition} We have:
	\begin{equation}\sum_{k=1}^\infty k \tilde{t}_k \left({\hat{x}}e^{\hat{y}}\right)^k-\hat{y} \Psi(e^x,\hbar) =0,
	\end{equation}	
	where $\hat x=x\cdot$ and $\hat y = \hbar D$.
\end{proposition}

\begin{proof}
	
	To obtain the KS operators for the generating function of double Hurwitz numbers we use the conjugation of the operators (\ref{KSbul}):
	\begin{align}\label{ksdob}
		a_{HH}&=e^{\frac{\beta}{2}m_0}a_\bullet\,e^{-\frac{\beta}{2}m_0}\\ \notag
		&=\sum_{k=1}^\infty k \tilde{t}_k \left(x\exp\left(\beta D\right)\right)^k-D\\
		&=\sum_{k=1}^\infty k \tilde{t}_k e^{\frac{\beta}{2}k(k-1)} x^ k\exp\left(\beta k D\right)-D, \notag
	\end{align}
	where $m_0:=\left(D-\frac{1}{2}\right)^2+\frac{1}{12}$ (the constant $\frac{1}{12}$ is not important for the calculations, but this way we get one of the standard generators of $w_{1+\infty}$, cf. the operator $\widehat{M}_0$ in Example~\ref{example:R1R2}).
	The KS operators \eqref{ksdob} act of the basis vectors as follows:
	\begin{equation*}
		a_{HH}\,\Phi_k^{HH}(z)=(k-1)\Phi_k^{HH}(z).
	\end{equation*}
	
	The operator $a_{HH}$ annihilates $\Phi_1^{HH}(x)$ and, therefore, describes the quantum spectral curve for this model. Namely, we have 
	\begin{equation*}
		A_{HH}\,\Psi(x,\hbar)=0
	\end{equation*}
	where
	\begin{equation}\label{qschh}
	A_{HH}=\sum_{k=1}^\infty k \tilde{t}_k e^{\frac{\hbar}{2}k(k-1)} x^ k\exp\left(\hbar k D\right)-\hbar D,
	\end{equation}
\end{proof}

\begin{remark}
	The wave function in this case is also given by the integral
	\begin{equation*}
		\Psi(x,\hbar)=\frac{e^{-\frac{\hbar}{8}}}{\sqrt{2\pi \hbar}}\int_{-\infty}^\infty dy\,\exp\left(-\frac{y^2}{2\hbar}-\frac{y}{2}+\sum_{k=1}^\infty \frac{\tilde{t}_k}{\hbar}(xe^y)^k\right).
	\end{equation*}
	considered as a formal series in $\tilde{t}_k$.
\end{remark}

\begin{remark}
	Particular specifications of $\tilde{\mathbf{t}}_k$ describe interesting examples of this model, in particular usual simple Hurwitz numbers \cite{AlexandrovTauFunctions2015}, triple Hodge integrals and string amplitude for the resolved conifold \cite{JianZhou2015-2}. Quantum spectral curves for all these examples are given by specifications of the more general expression \eqref{qschh}.
\end{remark}

\subsubsection{A particular example: one-parameter deformation of single Hurwitz numbers} Let us discuss an example of a particular specialization of double Hurwitz numbers given by $\tilde{t}_k={c}^{k-1}$, $k=1,2,\dots$. This gives a one-parameter deformation of single Hurwitz numbers considered in Example~\ref{example:one-parameter}. Up to a simple combinatorial factor, this is equivalent to the Hurwitz theory for the vector of blocks $(B^<_2,\dots,B^<_2,B^<_r)$ and $\nu=(1^{|\mu|})$ (recall that $B^<_r=B^|_r$ by Proposition~\ref{jucysonblocks}).

In this case the wave function is given by 
\begin{equation*}
	\Psi(x,\hbar)= \cor{ 
		\exp\left( \sum_{i=1}^\infty \frac {x^{i }\alpha_i}{i} \right) \exp \left(\hbar\cF_2\right) \exp\left(\sum_{i=1}^\infty  \frac{c^{i-1}\alpha_{-i}}{\hbar} \right)
	}.
\end{equation*}
Equation~\eqref{qschh} reduces to 
\begin{equation*} 
A_{HH}=\frac{{\hat{x}}e^{\hat{y}}}{(1-c\,{\hat{x}}e^{\hat{y}})^2}-\hat{y}.
\end{equation*}
Let us multiply this operator by $\frac{(1-c\,{\hat{x}}e^{\hat{y}})^2}{{\hat{x}}e^{\hat{y}}}$. The resulting equation for the wave function
\begin{equation*} 
\left(1-\left(e^{-\hat{y}}\hat{x}^{-1}-2c+c^2{\hat{x}}e^{\hat{y}}\right)\hat{y}\right)\Psi(x,\hbar)=0
\end{equation*}
describes the quantum spectral curve for this case.

\begin{remark}
	The restriction of the wave function $\Psi(x,\hbar)$ to $c=0$ is the wave function of the single Hurwitz numbers, and in this special case we recover the quantum spectral curve $e^{-\hat y} \hat x^{-1} \hat y -1$, which is equivalent to the one that was proved in this case in~\cite{Zhou2012}.
\end{remark}

This quantum spectral curve equation suggests that the spectral curve for the one-parameter family of Hurwitz numbers that we consider here should be 
\begin{equation}
ye^{-y} - (1+2cy)x + c^2y e^y x^2 =0,
\end{equation}
which is a deformation of the Lambert curve.



\bibliographystyle{abbrv}
\bibliography{HQKPref}

\end{document}